\documentclass[12pt,letterpaper]{article}
\usepackage{amsmath,amsfonts,amsthm,amssymb}

\theoremstyle{plain}

\numberwithin{equation}{section}
\newtheorem{thm}{Theorem}[section]
\newtheorem{lem}[thm]{Lemma}
\newtheorem{cor}[thm]{Corollary}

\newenvironment{exam}[1]%
{\begin{flushleft}\textbf{Example #1}.\enspace}%
{\end{flushleft}}

\newcounter{cond}

\newcommand{\complex}{{\mathbb C}}
\newcommand{\positive}{{\mathbb N}}
\newcommand{\real}{{\mathbb R}}
\newcommand{\integers}{{\mathbb Z}}
\newcommand{\ascript}{{\mathcal A}}
\newcommand{\bscript}{{\mathcal B}}
\newcommand{\pscript}{{\mathcal P}}

\newcommand{\cupdot}{\mathbin{\cup{\hskip-5.4pt}^\centerdot}\,}
\newcommand{\bigcupdotimplus}{{\bigcup _{i=1}^{m+1}}{\hskip-9pt}^\centerdot{\hskip 9pt}}

\newcommand{\brac}[1]{\left\{#1\right\}}
\newcommand{\paren}[1]{\left(#1\right)}
\newcommand{\sqbrac}[1]{\left[#1\right]}

\begin{document}

\title{AN ANHOMOMORPHIC LOGIC\\ FOR QUANTUM MECHANICS}
\author{Stan Gudder\\ Department of Mathematics\\
University of Denver\\ Denver, Colorado 80208\\
sgudder@math.du.edu}
\date{}
\maketitle

\begin{abstract}
Although various schemes for anhomomorphic logics for quantum mechanics have been considered in the past we shall mainly concentrate on the quadratic or grade-2 scheme. In this scheme, the grade-2 truth functions are called coevents. We discuss properties of coevents, projections on the space of coevents and the master observable. We show that the set of projections forms an orthomodular poset. We introduce the concept of precluding coevents and show that this is stronger than the previously studied concept of preclusive coevents. Precluding coevents are defined naturally in terms of the master observable. A result that exhibits a duality between preclusive and precluding coevents is given. Some simple examples are presented.
\end{abstract}

\section{Introduction}  
The study of anhomomorphic logics for quantum mechanics was initiated by R.~Sorkin \cite{sor94}. Since then, other investigations in the subject have been carried out \cite{gt09, gud09, sal02, sor07, sor09, sw08}. This work has usually been conducted in relation to the subject of quantum measure theory and is mainly motivated by the histories approach to quantum mechanics and quantum gravity \cite{gmh93, gr84, sor07}. One of the objectives of this subject is to describe the possible physical realities and to identify the actual physical reality.

The basic structure is given by a set of outcomes $\Omega$ together with an algebra $\ascript$ of subsets of
$\Omega$ whose elements are called events. It is generally agreed that possible realities are described by 1-0 functions from $\ascript$ to $\integers _2$ called truth functions. There are various schemes for choosing truth-functions that correspond to possible realities \cite{gt09, sor94}. Two of the most popular have been the linear and multiplicative schemes \cite{gt09, sor07, sor09}. We shall mainly concentrate on the quadratic scheme which has been rejected in the past but which we believe should be reconsidered. The elements of the chosen scheme are called coevents. Various methods have been devised for filtering out the unwanted coevents and selecting the actual reality. Three of these are called unitality, minimality and preclusivity \cite{gt09, sor07, sor09, sw08}.

In Section~2 we discuss the various schemes and give a reason for choosing the quadratic scheme. We call quadratic truth functions (grade-2) coevents. Section~3 discusses properties of these coevents. In Section~4 we consider projections on the space $\ascript ^*$ of coevents and observables. We show that the set of projections forms an orthomodular poset. We introduce the concept of the master observable and present its properties. Section~5 considers the concept of preclusivity. Preclusive coevents have already been discussed in the literature and we introduce a stronger notion that we call precluding coevents. This notion is defined naturally in terms of the master observable. We close this section with some simple examples and a result that exhibits a duality between preclusive and precluding coevents. For simplicity, the outcomes space $\Omega$ will be assumed to have finite cardinality.

\section{Truth Functions} 
Let $\Omega =\brac{\omega _1,\ldots ,\omega _n}$ be the sample space for some physical experiment or situation. We call the elements of $\Omega$ \textit{outcomes} and for simplicity we take $\Omega$ to be finite. The outcomes could correspond to particle locations or spin outcomes or particle trajectories, or fine-grained histories, etc. Subsets of $\Omega$ are called \textit{events} and we denote the set of all events $2^\Omega$ by $\ascript$. We use the notation $AB$ for the intersection $A\cap B$ and if $A\cap B=\emptyset$ we write $A\cupdot B=A\cup B$. We also use the notation $A+B$ for the symmetric difference $(AB')\cupdot (A'B)$ where $A'$ denotes the complement of $A$.

The logic for $\ascript$ gives the contact with reality; that is, the logic describes what actually happens. What actually happens may be determined by a truth function or 1-0 function $\phi\colon\ascript\to\integers _2$. If $\phi (A)=1$, then $A$ happens or $A$ is true and if $\phi (A)=0$, then $A$ does not happen or $A$ is false. Other terminology that is used is that $A$ occurs or does not occur. Now there are various admissible 1-0 functions depending on the situation or state of the system. For classical logic it is assumed that $\phi$ is a homomorphism. That is,
\begin{list} {(\arabic{cond})}{%
\setlength\itemindent{-7pt}}
\item [{(1)}]   
$\phi (\Omega )=1$\enspace (\textit{unital})
\item [{(2)}]   
$\phi (A+B)=\phi (A)+\phi (B)$\enspace (\textit{additive})
\item [{(3)}]   
$\phi (AB)=\phi (A)\phi (B)$\enspace (\textit{multiplicative})
\end{list}
Of course, in $\integers _2=\brac{0,1}$ addition is modulo $2$. If $\phi$ is a homomorphism, it can be shown that there exists an $\alpha\in\Omega$ such that $\phi (A)=1$ if and only if $\alpha\in A$. Defining the
\textit{containment map} $\alpha ^*\colon\ascript\to\integers _2$ by
\begin{equation*}
\alpha ^*(A)=
\begin{cases}1&\text{if $\alpha\in A$}\\ 0&\text{if $\alpha\notin A$}\end{cases}
\end{equation*}
we have that $\phi =\alpha ^*$. This is eminently reasonable for classical mechanics. For example a classical particle is definitely at a specific location $\alpha\in\Omega$ at any given time.

However, in quantum mechanics, assuming that $\phi$ must be a homomorphism can result in a contradiction. For example, consider a three-slit experiment where $\Omega =\brac{\omega _1,\omega _2,\omega _3}$ and
$\omega _i$ is the outcome that a quantum particle impinges the detection screen at a fixed small region $\Delta$ after going through slit $i$, $i=1,2,3$ \cite{sor07, sor09}. Then it is possible for
\begin{equation*}
\phi\paren{\brac{\omega _1,\omega _2}}=\phi\paren{\brac{\omega _2,\omega _3}}=0
\end{equation*}
If $\phi$ were a homomorphism, it follows that $\phi =0$; i.e., $\phi (A)=0$ for all $A\in\ascript$. Thus, nothing happens. Mathematically this gives a contradiction because by (1), $\phi (\Omega )=1$. This also gives a physical contradiction because there are certainly circumstances in which the particle is observed in $\Delta$. We again have that $\phi (\Omega )=1$.

The fundamental question becomes: What are the admissible 1-0 functions for quantum mechanics? We have seen that there are $n$ different homomorphisms corresponding to the $n$ classical states and we have argued that (1), (2) and (3) are too restrictive for quantum mechanics. On the other hand, there are $2^{2^n}$ possible 1-0 functions on
$\ascript$ and if we allow all of them, then the logic will have nothing to say. Thus, to have viable theory some restrictions must be put into place. In past studies, (1) is usually retained and either (2) or (3) are assumed
\cite{gt09, sor07, sor09}. In this work we shall not assume (1), (2) or (3) but shall postulate a generalization of (2). We shall also give an argument for the plausibility of this postulate. Since it would be unreasonable to consider a 1-0 function $\phi$ that satisfies $\phi (\emptyset )=1$ whenever we write $\phi\colon\ascript\to\integers _2$ we are assuming that $\phi (\emptyset )=0$.

But first it is instructive to examine the form of 1-0 functions that satisfy (2) or (3). If
$\phi ,\psi\colon\ascript\to\integers _2$ we define their \textit{sum} and \textit{product} by
$(\phi +\psi )(A)=\phi (A)+\psi (A)$ and $\phi\psi (A)=\phi (A)\psi (A)$. Of course, $\phi +\psi$ and $\phi\psi$ are again 1-0 functions. We can form polynomials in the containment maps $\alpha ^*$ for $\alpha\in\Omega$. For example
\begin{equation*}
\alpha ^*+\beta ^*+\alpha ^*\beta ^*+\alpha ^*\gamma ^*+\alpha ^*\beta ^*\gamma ^*
\end{equation*}
is a degree-3 polynomial. Since there are $2^{2^n}$ different polynomials, we conclude that every
$\phi\colon\ascript\to\integers _2$ can be uniquely represented by a polynomial (up to an ordering of the terms). The proof of parts of the following theorem are contained in \cite{gt09, sor07, sor09}. Also, this theorem and
Theorems~3.1, 3.3 and 3.6 are special cases of more general results in the field of combinatorial polarization
(\cite{dv09} and references therein). We include the proofs for the reader's convenience because they are shorter and more direct than the proofs for the more general results.

\begin{thm}       
\label{thm21}
{\rm (a)}\enspace A nonzero $\phi\colon\ascript\to\integers _2$ satisfies (2) if and only if
$\phi=\alpha _1^*+\cdots +\alpha _m^*$ for some $\alpha _1,\ldots ,\alpha _m\in\Omega$.
{\rm (b)}\enspace $\phi\colon\ascript\to\integers _2$ with $\phi\ne 0,1$ satisfies (3) if and only if
$\phi =\alpha _1^*\cdots\alpha _m^*$ for some $\alpha _1,\ldots ,\alpha _m\in\Omega$.
\end{thm}
\begin{proof}
(a)\enspace We first show that $\phi\colon\alpha\to\integers _2$ is additive if and only if $\phi$ satisfies
\begin{list} {(\arabic{cond})}{%
\setlength\itemindent{-7pt}}
\item [{(4)}]   
$\phi (A\cupdot B)=\phi (A)+\phi (B)$\enspace for all disjoint $A,B\in\ascript$.
\end{list}
If $\phi$ is additive, then clearly $\phi$ satisfies (4). Conversely, if $\phi$ satisfies (4) then for any $A,B\in\ascript$ we have
\begin{equation*}
\phi (A)=\phi (AB'\cupdot AB)=\phi (AB')+\phi (AB)
\end{equation*}
Therefore,
\begin{align*}
\phi (A+B)&=\phi (AB'\cupdot A'B)=\phi (AB')+\phi (A'B)\\
  &=\phi (AB')+\phi (AB)+\phi (AB)+\phi (A'B)\\
  &=\phi (A)+\phi (B)
\end{align*}
so $\phi$ is additive. Now suppose $\phi\colon\ascript\to\integers _2$ is additive and nonzero. Then there exist
$\alpha _1,\ldots ,\alpha _m\in\Omega$ such that $\phi (\alpha _i)=1$, $i=1,\ldots ,m$ and $\phi (\omega )=0$ for
$\omega\in\brac{\alpha _1,\ldots ,\alpha _m}'$ where for simplicity we write
$\phi (\omega )=\phi\paren{\brac{\omega}}$. By (4), for an $A\in\ascript$ we have
\begin{equation*}
\phi (A)=\sum _{\omega _i\in A}\phi (\omega _i)=\sum _{\alpha _i\in A}\phi (\alpha _i)
  =\sum _{\alpha _i\in A}1=\sum _{i=1}^m\alpha _i^*(A)
\end{equation*}
Hence, $\phi =\alpha _1^*+\cdots +\alpha _m^*$ and the converse is clear.\newline
(b)\enspace If $\phi =\alpha _1^*\cdots\alpha _m^*$, then
\begin{align*}
\phi (AB)&=\alpha  _1^*(AB)\cdots\alpha _m^*(AB)=\alpha _1^*(A)\alpha _1^*(B)\cdots\alpha _m^*(A)\alpha _m^*(B)\\
  &=\alpha _1^*(A)\cdots\alpha _m^*(A)\alpha _1^*(B)\cdots\alpha _m^*(B)=\phi (A)\phi (B)
\end{align*}
so $\phi$ is multiplicative. Conversely, suppose that $\phi$ is multiplicative and $\phi\ne 0,1$. If $A\subseteq B$ we have $AB=A$ so that
\begin{equation*}
\phi (A)=\phi (AB)=\phi (A)\phi (B)\le\phi (B)
\end{equation*}
Since $\phi\ne 0$ there exists an $A\in\ascript$ with $\phi (A)=1$. Let
\begin{equation*}
B=\cap\brac{A\in\ascript\colon\phi (A)=1}
\end{equation*}
Then $B$ is the smallest set with $\phi (B)=1$; that is, $\phi (A)=1$ if and only if $B\subseteq A$. Since $\phi\ne 1$,
$B\ne\emptyset$. Letting $B=\brac{\alpha _1,\ldots ,\alpha _m}$ we have that $\phi (A)=1$ if and only if
$\alpha _i\in A$, $i=1,\ldots ,m$. Hence,
\begin{equation*}
\phi (A)=\alpha _1^*\cdots\alpha _m^*(A)\qedhere
\end{equation*}
\end{proof}

It follows from Theorem~\ref{thm21} that $\phi\colon\ascript\to\integers _2$ is a homomorphism if and only if
$\phi =\alpha ^*$ for some $\alpha\in\Omega$. We now consider a generalization of the additivity condition~(2). For
$\phi\colon\ascript\to\integers _2$ we define the $m$-\textit{point interference} $I_\phi ^m$ as the map from
$m$-tuples of distinct elements of $\Omega$ into $\integers _2$ given by
\begin{equation*}
I_\phi ^m(\alpha _1,\ldots ,\alpha _m)=\phi\paren{\brac{\alpha _1,\ldots ,\alpha _m}}
  +\phi (\alpha _1)+\cdots +\phi (\alpha _m)
\end{equation*}
where $m\in\positive$ with $m\ge 2$. Since it is clear that $\phi$ is additive if and only if $I_\phi ^m=0$ for all $m$ with $2\le m\le n$, we see that $I_\phi ^m$ gives a measure of the amount that $\phi$ deviates from being additive. An analogous definition is used to describe interference for quantum measures \cite{gud09}. Our basic postulate is that $m$-point interference is governed by two-point interferences, in the sense that
\begin{equation}         
\label{eq21}
I_\phi ^m(\alpha _1,\ldots ,\alpha _m)=\sum _{i<j=1}^mI_\phi ^2(\alpha _i,\alpha _i)
\end{equation}
We call \eqref{eq21} for all $2\le m\le n$ the \textit{two-point interference condition}.

We say that $\phi\colon\ascript\to\integers _2$ is \textit{grade}-2 \textit{additive} if 
\begin{equation*}
\phi (A\cupdot B\cupdot C)=\phi (A\cupdot B)+\phi (A\cupdot C)+\phi (B\cupdot C)+\phi (A)+\phi (B)+\phi (C)
\end{equation*}
for all mutually disjoint $A,B,C\in\ascript$. We also call (2) \textit{grade}-1 \textit{additivity} and clearly grade-1 additivity implies grade-2 additivity but we shall see that the converse does not hold. One can also define higher grade additivities but these shall not be considered here \cite{sal02}. 

\begin{thm}       
\label{thm22}
A function $\phi\colon\ascript\to\integers _2$ is grade-2 additive if and only if $\phi$ satisfies the two-point interference condition.
\end{thm}
\begin{proof}
We shall show in Corollary~3.2 that $\phi$ is grade-2 additive if and only if
\begin{equation}     
\label{eq22}
\phi\paren{\brac{\alpha _1,\ldots ,\alpha _m}}=\sum _{i<j=1}^m\phi\paren{\brac{\alpha _i,\alpha _j}}
   +\tfrac{1}{2}\sqbrac{1-(-1)^m}\sum _{i=1}^m\phi (\alpha _i)
\end{equation}
for all $m\in\positive$ with $2\le m\le n$. But \eqref{eq22} is equivalent to
\begin{align}     
\label{eq23}
I_\phi ^m&(\alpha _1,\ldots ,\alpha _m)+\sum _{i=1}^m\phi (\alpha _i)\notag\\
  &=\sum _{i<j=1}^mI_\phi ^2(\alpha _i,\alpha _j)+(m-1)\sum _{i=1}^m\phi (\alpha _j)
  +\tfrac{1}{2}\sqbrac{1-(-1)^m}\sum _{i=1}^m\phi (\alpha _i)
\end{align}
Moreover, \eqref{eq23} is equivalent to
\begin{align*}
I_\phi ^m(\alpha _1,\ldots ,\alpha _m)
  &=\sum _{i<j=1}^mI_\phi ^2(\alpha _i,\alpha _j)+\sqbrac{m+\tfrac{1}{2}\paren{1-(-1)^m}}
  \sum _{i=1}^m\phi (\alpha _i)\\
   &=\sum _{i<j=1}^mI_\phi ^2(\alpha _i,\alpha _j)
\end{align*}
which is the two-point interference condition.
\end{proof}

\section{Grade-2 Additivity} 
The two-point interference condition is analogous to an interference condition that holds for quantum measures \cite{gud09} and in our opinion this condition should hold for all (finite) quantum systems. It follows from
Theorem~\ref{thm22} that the set of possible realities for a quantum system is described by the set $\ascript ^*$ of grade-2 additive functions from $\ascript$ to $\integers _2$. We call the elements of $\ascript ^*$ coevents. We first give the result that was needed in the proof of Theorem~\ref{thm22}.

\begin{thm}       
\label{thm31}
A map $\phi\colon\ascript\to\integers _2$ is a coevent if and only if $\phi$ satisfies
\begin{equation}     
\label{eq31}
\phi\paren{A_1\cupdot\cdots\cupdot A_m}=\sum _{i<j=1}^m\phi (A_i\cupdot A_j)
  +\tfrac{1}{2}\sqbrac{1-(-1)^m}\sum _{i=1}^m\phi (A_i)
\end{equation}
for all $m\in\positive$ with $2\le m\le n$.
\end{thm}
\begin{proof}
If \eqref{eq31} holds, then $\phi$ is clearly a coevent. Conversely, assume that $\phi$ is a coevent. We now prove \eqref{eq31} by induction on $m$. The result holds for $m=2,3$. Suppose the result holds for $m\ge 2$, where $m$ is odd. Then
\begin{align*}
\phi\paren{\bigcupdotimplus A_i}&=\phi\sqbrac{A_1\cupdot\cdots\cupdot\paren{A_m\cupdot A_{m+1}}}\\
  &=\sum _{i<j=1}^{m-1}\phi\paren{A_i\cupdot A_j}
  +\sum _{i=1}^{m-1}\phi\sqbrac{A_i\cupdot\paren{A_m\cupdot A_{m+1}}}\\
  &\quad +\sum _{i=1}^{m-1}\phi (A_i)+\phi\paren{A_m\cupdot A_{m+1}}\\
  &=\sum _{i<j=1}^{m-1}\phi\paren{A_i\cupdot A_j}+\sum _{i=1}^{m-1}\phi\paren{A_i\cupdot A_m}\\
  &\quad +\sum _{i=1}^{m-1}\phi\paren{A_i\cupdot A_{m+1}}+\phi\paren{A_m\cupdot A_{m+1}}\\
  &=\sum _{i<j=1}^{m+1}\phi\paren{A_i\cupdot A_j}
\end{align*}
Suppose the result holds for $m\ge 2$ where $m$ is even. Then
\begin{align*}
\phi\paren{\bigcupdotimplus A_i}&=\phi\sqbrac{A_1\cupdot\cdots\cupdot\paren{A_m\cupdot A_{m+1}}}\\
  &=\sum _{i<j=1}^{m-1}\phi\paren{A_i\cupdot A_j}
  +\sum _{i=1}^{m-1}\phi\sqbrac{A_i\cupdot\paren{A_m\cupdot A_{m+1}}}\\
  &=\sum _{i<j=1}^{m-1}\phi \paren{A_i\cupdot A_j}+\sum _{i=1}^{m-1}\phi\paren{A_i\cupdot A_m}
  +\sum _{i=1}^{m-1}\phi\paren{A_i\cupdot A_{m+1}}\\
  &\quad +\phi\paren{A_m\cupdot A_{m+1}}+\sum _{i=1}^{m-1}\phi (A_i)+\phi (A_m)+\phi (A_{m+1})\\
  &=\sum _{i<j=1}^{m+1}\phi\paren{A_i\cupdot A_j}+\sum _{i=1}^{m+1}\phi (A_i)
\end{align*}
The result now follows by induction.
\end{proof}

\begin{cor}       
\label{cor32}
A map $\phi\colon\ascript\to\integers _2$ is a coevent if and only if \eqref{eq22}  holds for all $2\le m\le n$.
\end{cor}
\begin{proof}
If $\phi$ is a coevent, then $\phi$ satisfies \eqref{eq22} by letting $A_i=\brac{\alpha _i}$ in \eqref{eq31}. Conversely, suppose $\phi$ satisfies \eqref{eq22} and let $A,B,C\in\ascript$ be mutually disjoint with
$A=\brac{\alpha _1,\ldots ,\alpha _n}$, $B=\brac{\beta _1,\ldots ,\beta _s}$, $C=\brac{\gamma _1,\ldots ,\gamma _t}$. The special cases in which at least one of the sets $A$, $B$ or $C$ has cardinality less than two are easily treated so we assume their cardinalities are at least two. Then by \eqref{eq22} we have that
\begin{align*}
\phi&\paren{A\cupdot B}+\phi\paren{A\cupdot C}+\phi\paren{B\cupdot C}+\phi (A)+\phi (B)+\phi (C)\\
  &=\sum _{i=1}^r\sum _{j=1}^s\phi\paren{\brac{\alpha _i,\beta _j}}
  +\tfrac{1}{2}\sqbrac{1-(-1)^{r+s}}\sqbrac{\sum _{i=1}^r\phi (\alpha _i)+\sum _{i=1}^s\phi (\beta _i)}\\
  &\quad +\sum _{i=1}^r\sum _{j=1}^t\phi\paren{\brac{\alpha _i,\gamma _j}}
  +\tfrac{1}{2}\sqbrac{1-(-1)^{r+t}}\sqbrac{\sum _{i=1}^r\phi (\alpha _i)+\sum _{i=1}^t\phi (\gamma _i)}\\
  &\quad +\sum _{i=1}^s\sum _{j=1}^t\phi\paren{\brac{\beta _i,\gamma _j}}
  +\tfrac{1}{2}\sqbrac{1-(-1)^{s+t}}\sqbrac{\sum _{i=1}^s\phi (\beta _i)+\sum _{i=1}^t\phi (\gamma _i)}\displaybreak\\
  &\quad +\sum _{i<j=1}^r\phi\paren{\brac{\alpha _i,\alpha _j}}
  +\tfrac{1}{2}\sqbrac{1-(-1)^r}\sum _{i=1}^r\phi (\alpha _i)+\sum _{i<j=1}^s\phi\paren{\brac{\beta _i,\beta _j}}\\
  &\quad +\tfrac{1}{2}\sqbrac{1-(-1)^s}\sum _{i=1}^s\phi (\beta _i)
  +\sum _{i<j=1}^t\phi\paren{\brac{\gamma _i,\gamma _j}}+\tfrac{1}{2}\sqbrac{1-(-1)^t}\sum _{i=1}^t\phi (\gamma _i)\\
  &=\sum _{i<j=1}^r\phi\paren{\brac{\alpha _i,\alpha _j}}+\sum _{i<j=1}^s\phi\paren{\brac{\beta _i,\beta _j}}
  +\sum _{i<j=1}^t\phi\paren{\brac{\gamma _i,\gamma _j}}\\
  &\quad +\sum _{i=1}^r\sum _{j=1}^s\phi\paren{\brac{\alpha _i,\beta _j}}
  +\sum _{i=1}^r\sum _{j=1}^t\phi\paren{\brac{\alpha _i,\gamma _j}}
  +\sum _{i=1}^s\sum _{j=1}^t\phi\paren{\brac{\beta _i,\gamma _j}}\\
  &\quad +\tfrac{1}{2}\sqbrac{1-(-1)^{r+s+t}}
  \sqbrac{\sum _{i=1}^r\phi (\alpha _i)+\sum _{i=1}^s\phi (\beta _i)+\sum _{i=1}^t\phi (\gamma _i)}\\
  &=\phi\paren{A\cupdot B\cupdot C}
\end{align*}
where the second to last equality comes from the fact that in $\integers _2$ we have
\begin{align*}
&\tfrac{1}{2}\sqbrac{1-(-1)^{r+s}}(a+b)+\tfrac{1}{2}\sqbrac{1-(-1)^{r+t}}(a+c)+\tfrac{1}{2}\sqbrac{1-(-1)^{s+t}}(b+c)\\
  &\quad +\tfrac{1}{2}\sqbrac{1-(-1)^r}a+\tfrac{1}{2}\sqbrac{1-(-1)^s}b+\tfrac{1}{2}\sqbrac{1-(-1)^t}c\\
  &=\tfrac{1}{2}\sqbrac{1-(-1)^{r+s+t}}(a+b+c)
\end{align*}
for all $r,s,t\in\positive$, $a,b,c\in\integers _2$ which can be checked by cases.
\end{proof}

We call the set of coevents $\ascript ^*$ an \textit{anhomomorphic logic}. Various schemes for anhomomorphic logics have been developed in the literature \cite{gt09, sor07, sor09}. In fact, the present scheme was rejected in \cite{gt09} because in some examples there were not enough coevents available. The reason for this is that only minimal (or primitive) and unital coevents were considered. We disagree with this analysis and believe that these restrictions are completely unnecessary. However, we shall later consider another means for restricting coevents that has already been used, called preclusivity. We now give further properties of coevents.

\begin{thm}       
\label{thm33}
A map $\phi\colon\ascript\to\integers _2$ is a coevent for $\Omega =\brac{\omega _1,\ldots ,\omega _n}$ if and only if $\phi$ is a first or second degree polynomial in the $\omega _i^*$, that is
\begin{equation}     
\label{eq32}
\phi =\sum _{i=1}^na_i\omega _i^*+\sum _{i,j=1}^nb_{ij}\omega _i^*\omega _j^*
\end{equation}
where $a_i,b_{ij}\in\integers _2$.
\end{thm}
\begin{proof}
It is easy to check that $\omega _i^*\omega _j^*$ are coevents and that the sum of coevents is a coevent. Hence, any map $\phi\colon\ascript\to\integers _2$ of the form \eqref{eq32} is a coevent. Conversely, suppose
$\phi\colon\ascript\to\integers _2$ is a coevent. Reorder the $\omega _i$ if necessary so that
\begin{align*}
&\phi(\omega _1)=\cdots =\phi (\omega _r)=1,\phi\paren{\brac{\omega _i,\omega _j}}
  =\cdots =\phi\paren{\brac{\omega _{i'},\omega _{j'}}}=1, i,j,i',j'\le r\\
  &\phi\paren{\brac{\omega _s,\omega _t}}=\cdots =\phi\paren{\brac{\omega _{s'},\omega _{t'}}}=1, s,s'\le r,\ t,t'>r\\
  &\phi\paren{\brac{\omega _u,\omega _v}}=\cdots =\phi\paren{\brac{\omega _{u'},\omega _{v'}}}=1, u,v,u',v'>r
\end{align*}
and $\phi$ is $0$ for all other singleton and doubleton sets. Define $\psi\colon\ascript\to\integers _2$ by
\begin{equation*}
\psi =\sum _{k=1}^r\omega _k^*+\omega _i^*\omega _j^*+\cdots +\omega _{i'}^*\omega _{j'}^*
  +\omega _u^*\omega _v^*+\cdots +\omega _{u'}^*\omega _{v'}^*
  +\sum _{k=1}^r\sum _{w\in W}\omega _k^*\omega _w^*
\end{equation*}
where $W$ is the set of indices that are not represented above. Then $\phi$ and $\psi$ are coevents that agree on singleton and doubleton sets. By \eqref{eq22} $\phi$ and $\psi$ coincide.
\end{proof}

We now illustrate Theorem~\ref{thm33} with an example. Let $\Omega =\brac{\omega _1,\ldots ,\omega _5}$ and suppose $\phi\in\ascript ^*$ satisfies $\phi (\omega _1)=\phi (\omega _2)=1$,
\begin{equation*}
\phi\paren{\brac{\omega _1,\omega _2}}=\phi\paren{\brac{\omega _2,\omega _3}}
  =\phi\paren{\brac{\omega _4,\omega _5}}=1
\end{equation*}
and $\phi$ is $0$ for all other singleton and doubleton sets. Define $\psi\in\ascript ^*$ by
\begin{equation*}
\psi=\omega _1^*+\omega _2^*+\omega _1^*\omega _2^*+\omega _4^*\omega _5^*+\omega _1^*\omega _3^*
  +\omega _1^*\omega _4^*+\omega _1^*\omega _5^*+\omega _2^*\omega _4^*+\omega _2^*\omega _5^*
\end{equation*}
Then $\phi$ and $\psi$ are coevents that agree on singleton and doubleton sets so by \eqref{eq22} $\phi$ and
$\psi$ coincide.

The next result follows from the proof of Theorem~\ref{thm33}.

\begin{cor}       
\label{cor34}
Given any assignment of zeros and ones to the singleton and doubleton sets of
$\Omega =\brac{\omega _1,\ldots ,\omega _n}$, there exists a unique coevent $\phi\colon\ascript\to\integers _2$ that has these values.
\end{cor}

It follows from Theorem~\ref{thm33} that the anhomomorphic logic $\ascript ^*$ is a vector space over $\integers _2$ with dimension
\begin{equation*}
\dim (\ascript ^*)=n+\binom{n}{2}=\frac{n(n+1)}{2}
\end{equation*}
Hence, the cardinality of $\ascript ^*$ is $2^{n(n+1)/2}$ which is much smaller than the cardinality $2^{2^n}$ of the set of all 1-0 functions on $\Omega$.

\begin{lem}       
\label{lem35}
{\rm (a)}\enspace $\phi\colon\ascript\to\integers _2$ is grade-1 additive if and only if
$\phi (A\cup B)=\phi (A)+\phi (B)+\phi (AB)$ for all $A,B\in\ascript$ and $\phi (\emptyset )=0$.
{\rm (b)}\enspace $\phi\colon\ascript\to\integers _2$ is grade-2 additive if and only if
\begin{equation}     
\label{eq33}
\phi (A\cup B)=\phi (A)+\phi (B)+\phi (AB)+\phi (A+B)+\phi (AB')+\phi (A'B)
\end{equation}
for all $A,B\in\ascript$.
\end{lem}
\begin{proof}
(a)\enspace If $\phi$ is grade-1 additive, then
\begin{equation*}
\phi (A\cupdot B)=\phi (A+B)=\phi (A)+\phi (B)
\end{equation*}
Hence,
\begin{equation*}
\phi (A)=\phi\sqbrac{(AB)\cupdot (AB')}=\phi (AB)+\phi (AB')
\end{equation*}
for all $A,B\in\ascript$. We conclude that
\begin{align*}
\phi (A\cup B)&=\phi (AB')+\phi (A'B)+\phi (AB)\\
  &=\phi (A)+\phi (AB)+\phi (B)+\phi (AB)+\phi (AB)\\
  &=\phi (A)+\phi (B)+\phi (AB)
\end{align*}
Also, it is clear that $\phi (\emptyset )=0$. Conversely, suppose the given formulas hold. Then as before
\begin{align*}
\phi (A+B)&=\phi (AB')+\phi (A'B)=\phi (A)+\phi (AB)+\phi (B)+\phi (AB)\\
&=\phi (A)+\phi (B)
\end{align*}
so $\phi$ is grade-1 additive. (b)\enspace If $\phi$ is grade-2 additive, then
\begin{align*}
\phi (A\cup B)&=\phi\sqbrac{(AB)\cupdot (AB')\cupdot (A'B)}\\
  &=\phi (A+B)+\phi (A)+\phi (B)+\phi (AB)+\phi (AB')+\phi (A'B)
\end{align*}
Conversely, if \eqref{eq33} holds, then letting $A_1=A\cupdot C$, $B_1=B\cupdot C$ we have that
\begin{align*}
\phi&(A\cupdot B\cupdot C)=\phi (A_1\cup B_1)\\
  &=\phi (A_1+B_1)+\phi (A_1)+\phi (B_1)+\phi (A_1B'_1)+\phi (A'_1B_1)+\phi (A_1B_1)\\
  &=\phi (A\cupdot B)+\phi (A\cupdot C)+\phi (B\cupdot C)+\phi (A)+\phi (B)+\phi (C)
\end{align*}
which is grade-2 additivity.
\end{proof}

In the next result, $\ascript\times\ascript$ denotes the collection $2^{\Omega\times\Omega}$ of all subsets of
$\Omega\times\Omega$. It is easy to see that the map $\lambda$ in this result is not unique.

\begin{thm}       
\label{thm36}
$\phi\colon\ascript\to\integers _2$ is a coevent if and only if there exists a grade-1 additive map
$\lambda\colon\ascript\times\ascript\to\integers _2$ such that $\phi (A)=\lambda (A\times A)$ for all $A\in\ascript$.
\end{thm}
\begin{proof}
If $\phi\colon\ascript\to\integers _2$ is a coevent, then by Theorem~\ref{thm33}, $\phi$ has the form
\begin{equation*}
\phi =\sum\alpha _i^*+\sum\beta _i^*\gamma _i^*
\end{equation*}
for $\alpha _i,\beta _i,\gamma _i\in\Omega$. Define $\lambda\colon\ascript\times\ascript\to\integers _2$ by
\begin{equation*}
\lambda =\sum (\alpha _i\times\alpha _i)^*+\sum (\beta _i\times\gamma _j)^*
\end{equation*}
Then by Theorem~\ref{thm21}(a), $\lambda$ is grade-1 additive and by Lemma~\ref{lem35}(a),
$\lambda (A\times A)=\phi (A)$ for all $A\in\ascript$. Conversely, suppose
$\lambda\colon\ascript\times\ascript\to\integers _2$ is grade-1 additive. By Theorem~\ref{thm21}(a), $\lambda$ has the form
\begin{equation*}
\lambda =\sum (\alpha _i\times\beta _j)^*
\end{equation*}
If $\phi\colon\ascript\to\integers _2$ satisfies $\phi (A)=\lambda (A\times A)$, then
\begin{align*}
\phi (A)&=\sum (\alpha _i\times\beta _j)^*(A\times A)=\sum\alpha _i^*(A)\beta _j^*(A)\\
  &=\sum\alpha _i^*(A)+\sum (\alpha _i^*\beta _j^*)(A)
\end{align*}
where the first summation on the right side is when $\alpha _i=\beta _j$. It follows from Theorem~\ref{thm21}(b) that
$\phi$ is a coevent.
\end{proof}

We now briefly discuss the possible strange behavior of coevents. Taking the particle location interpretation, the ``superposition'' $\omega _1^*+\omega _2^*$ states that the particle is at position~1 and at position~2 but not at $1$ or $2$. The ``entanglement'' $\omega _1^*\omega _2^*$ states that the particle is at position ~1 or at position~2 but if we look closely, it is not at either $1$ or $2$.

\section{Projections and Observables} 
In the sequel, $\Omega =\brac{\omega _1,\ldots ,\omega _n}$ is a finite set, $\ascript$ is the Boolean algebra of all subsets of $\Omega$ and $\ascript ^*$ is the anhomomorphic logic. We have seen in Section~3 that $\ascript^*$ is a $n(n+1)/2$ dimensional vector space over $\integers _2$ with basis consisting of the additive terms $\omega _i^*$ and the quadratic terms $\omega _i^*\omega _j^*$. A \textit{projection} on $\ascript ^*$ is a linear (or additive) idempotent map $P\colon\ascript ^*\to\ascript ^*$. That is, $P(\phi +\psi )=P\phi +P\psi$ for all
$\phi ,\psi\in\ascript ^*$ and $P^2=PP=P$. We denote the set of all projections on $\ascript ^*$ by
$\pscript (\ascript ^*)$. If $P,Q\in\pscript (\ascript ^*)$ with $PQ=QP$, then it is clear that $P+Q$ and $PQ$ are again projections. For $P\in \pscript (\ascript ^*)$ we define $P'\in\pscript (\ascript ^*)$ by $P'=I+P$. For
$P,Q\in\pscript (\ascript ^*)$ we define $P\le Q$ if $PQ=QP=P$. We call a partially ordered set a \textit{poset}. The greatest lower bound and least upper bound (if they exist) in a poset are denoted by $P\wedge Q$ and $P\vee Q$, respectively. For related work we refer the reader to \cite{ube09}

\begin{thm}       
\label{thm41}
{\rm (a)}\enspace $\paren{\pscript (\ascript ^*),\le}$ is a poset.
{\rm (b)}\enspace For $P,Q\in\pscript (\ascript ^*)$ we have that $P''=P$, $P\wedge P'=0$ and $P\le Q$ implies
$Q'\le P'$.
{\rm (c)}\enspace If $PQ=QP$ then $P\wedge Q=PQ$ and $P\vee Q=P+Q+PQ$.
\end{thm}
\begin{proof}
(a)\enspace Clearly, $P\le P$ for all $P\in\pscript (\ascript ^*)$. If $P\le Q$ and $Q\le P$, then
\begin{equation*}
P=PQ=QP=Q
\end{equation*}
If $P\le Q$ and $Q\le R$ then
\begin{align*}
PR&=PQR=PQ=P\\
\intertext{and}
RP&=RQP=QP=P
\end{align*}
Hence, $P\le R$ so $\paren{\pscript (\ascript ^*),\le}$ is a poset.
(b)\enspace Clearly $P''=P$. If $P\le Q$, then
\begin{equation*}
(I+P)(I+Q)=I+P+Q+PQ=I+P+Q+P=I+Q
\end{equation*}
Similarly, $(I+Q)(I+P)=I+Q$ so $Q'\le P'$. If $Q\le P,P'$, then
\begin{equation*}
Q=QP'=Q(I+P)=Q+QP=Q+Q=0
\end{equation*}
Hence, $P\wedge P'=0$.
(c)\enspace Since
\begin{equation*}
(PQ)P=(QP)P=QP=PQ
\end{equation*}
we have that $PQ\le P$ and similarly $PQ\le Q$. Suppose that $R\in\pscript (\ascript ^*)$ with $R\le P,Q$. Then
\begin{equation*}
RPQ=RQ=R
\end{equation*}
so that $R\le PQ$. Hence, $P\wedge Q=PQ$. By DeMorgan's law we have that
\begin{align*}
P\vee Q&=(P'\wedge Q')'=(P'Q')'=I+(I+P)(I+Q)\\
  &=I+I+P+Q+PQ=P+Q+PQ\qedhere
\end{align*}
\end{proof}

A poset $(\pscript ,\le )$ with a mapping ${}'\colon\pscript\to\pscript$ satisfying the conditions of
Theorem~\ref{thm41}(b) is called an \textit{orthocomplemented poset}. If $P\le Q'$ we write $P\perp Q$ and say that $P$ and $Q$ are \textit{orthogonal}. Of course, $P\perp Q$ if and only if $Q\perp P$. An orthocomplemented poset
$(\pscript ,\le , {}'\,)$ is called an \textit{orthomodular poset} if for $P,Q\in\pscript$ we have that $P\perp Q$ implies
$P\vee Q$ exists and $P\le Q$ implies
\begin{equation*}
Q=P\vee (Q\wedge P')
\end{equation*}
 
\begin{thm}       
\label{thm42}
{\rm (a)}\enspace For $P,Q\in\pscript (\ascript ^*)$, $P\perp Q$ if and only if $PQ=QP=0$.
{\rm (b)}\enspace $\paren{\pscript (\ascript ^*),\le , {}'\,}$ is an orthomodular poset.
\end{thm}
\begin{proof}
(a)\enspace If $P\perp Q$, then
\begin{equation*}
P=P(I+Q)=P+PQ
\end{equation*}
Adding $P$ to both sides gives $PQ=0$. Similarly, $QP=0$. If $PQ=QP=0$, then
\begin{equation*}
P(I+Q)=P+PQ=P
\end{equation*}
Similarly, $(I+Q)P=P$ so $P\le Q'$.
(b)\enspace If $P\perp Q$, then by (a) we have that $PQ=QP=0$. Hence, by Theorem~\ref{thm41}(c) we conclude that $P\vee Q$ exists and $P\vee Q=P+Q$. Now assume that $P\le Q$. Since $Q'\le P'$ we have that $Q'\perp P$. Hence, as before $P\vee Q'$ exists. It follows that $Q\wedge P'=(P\vee Q')'$ exists. Since
\begin{equation*}
P\le P\vee Q'=(Q\wedge P')'
\end{equation*}
we have that $P\perp Q\wedge P'$. Hence, $P\vee (Q\wedge P')=P+Q\wedge P'$ exists. By Theorem~\ref{thm41}(c) we have that
\begin{equation*}
Q\wedge P'=QP'=Q(I+P)=Q+PQ
\end{equation*}
Therefore
\begin{equation*}
Q=P+(Q+P)=P+Q+PQ=P+Q\wedge P'=P\vee (Q\wedge P')
\end{equation*}
It follows that $\paren{\pscript (\ascript ^*),\le ,{}'\,}$ is an orthomodular poset.
\end{proof}

An orthomodular poset is frequently called a ``quantum logic.'' Quantum logics have been studied for over 45 years in the foundations of quantum mechanics \cite{bvn36, dp00, gud88, jau68, mac63, pir68, var68}. It is interesting that the present formalism is related to this older approach. In the quantum logic approach the elements of
$\pscript (\ascript ^*)$ are thought of as quantum propositions or events. When we later consider observables we shall see that there is a natural correspondence between elements of $\ascript$ and some of the elements of
$\pscript (\ascript ^*)$. These elements of $\pscript (\ascript ^*)$ then become quantum generalizations of the events in $\ascript$. In accordance with the quantum logic approach we say that $P,Q\in\pscript (\ascript ^*)$ are
\textit{compatible} if there exist mutually orthogonal elements $P_1,Q_1,R\in\pscript (\ascript ^*)$ such that $P=R_1\vee R$ and $Q=Q_1\vee R$. Compatible events describe events that can occur in a single measurement
\cite{jau68, mac63, pir68}.

\begin{thm}       
\label{thm43}
$P,Q\in\pscript (\ascript ^*)$ are compatible if and only if $PQ=QP$.
\end{thm}
\begin{proof}
If $P,Q$ are compatible, there exist $P_1$, $Q_1$ and $R\in\pscript (\ascript ^*)$ satisfying the given conditions. Then
$P=P_1+R$, $Q=Q_1+R$ so by Theorem~\ref{thm42}(a) we have that
\begin{equation*}
PQ=(P_1+R)(Q_1+R)=P_1Q_1+P_1R+RQ_1+R=R
\end{equation*}
Similarly, $QP=R$. Conversely, suppose that $PQ=QP$. Define $R=PQ$, $P_1=P+PQ$, $Q_1=Q+PQ$. It is easy to check that $P_1$, $Q_1$ and $R$ are mutually orthogonal elements of $\pscript (\ascript ^*)$. Applying
Theorem~\ref{thm41}(c) we conclude that
\begin{align*}
P&=(P+PQ)+PQ=P_1+R=P_1\vee R\\
\intertext{and}
Q&=(Q+PQ)+PQ=Q_1+R=Q_1\vee R
\end{align*}
Hence, $P$ and $Q$ are compatible.
\end{proof}

In the quantum logic approach, instead of $\pscript (\ascript ^*)$ we frequently have the projective space
$\pscript (H)$ of orthogonal (self-adjoint) projections on a complex Hilbert space $H$. For $P, Q\in\pscript (H)$ we define $P\le Q$ if $P=PQ$. It is well-known that Theorems~\ref{thm42} and \ref{thm43} hold for
$\paren{\pscript (H),\le}$. However, it does not immediately follow that these theorems hold for
$\paren{\pscript (\ascript ^*),\le}$ because the structure of the vector space $\ascript ^*$ over $\integers _2$ is quite different than that of an inner product space over the complex field $\complex$. Also, $\pscript (\ascript ^*)$ consists of all projections on $\ascript ^*$ while $\pscript (H)$ consists of only orthogonal projections. This is illustrated in Example~1 at the end of this section.

For further emphasis we give some examples of the differences between $\pscript (\ascript ^*)$ and $\pscript (H)$. For commuting projections $P,Q\in\pscript (\ascript ^*)$ we have $P+Q\in\pscript (\ascript ^*)$ which is not true in
$\pscript (H)$. For $P,Q\in\pscript (H)$ if $PQ=0$ then $PQ=QP$ which is not true in $\pscript (\ascript ^*)$ as is shown in Example~1 at the end of this section. Theorem~\ref{thm41}(c) does not hold in
$\paren{\pscript (H),\le}$. Finally, it is known that $\paren{\pscript (H),\le}$ is a lattice ($P\wedge Q$ and $P\vee Q$ always exist). However, it is not known whether $\pscript (\ascript ^*)$ is a lattice and this would be an interesting problem to investigate.

We have seen that $\brac{\omega _i^*,\omega _i^*\omega _j^*\colon i,j=1,\ldots ,n}$ forms a basis for the vector space $\ascript ^*$. For $\omega _i\in\Omega$ define the map $P(\omega _i)\colon\ascript ^*\to\ascript ^*$ by
$P(\omega _i)\omega _j^*=\omega _i^*\omega _j^*$,
\begin{equation*}
P(\omega _i)\omega _i^*\omega _j^*=P(\omega _i)\omega _j^*\omega _i^*=\omega _i^*\omega _j^*
\end{equation*}
and for $i,j,k$ distinct $P(\omega _i)\omega _j^*\omega _k^*=0$ and extended $P(\omega _i)$ to $\ascript ^*$ by linearity. It is easy to check that $P(\omega _i)\in\pscript (\ascript ^*)$, $i=1,\ldots ,n$. Moreover, one can check that
$P(\omega _i)+P(\omega _j)\in\pscript (\ascript ^*)$ and that
$P(\omega _i)P(\omega _j)=P(\omega _j)P(\omega _i)\in\pscript (\ascript ^*)$. For $A\in\ascript$ define
$P(A)\colon\ascript ^*\to\ascript ^*$ by
\begin{equation*}
P(A)=\sum\brac{P(\omega _i)+P(\omega _i)P(\omega _j)\colon\omega _i,\omega _j\in A, i<j}
\end{equation*}
It follows that $P(A)\in\pscript (\ascript ^*)$ for all $A\in\ascript$. For example,
\begin{align*}
P&\paren{\brac{\omega _1,\omega _2,\omega _3}}\\
  &=P(\omega _1)+P(\omega _2)+P(\omega _3)+P(\omega _1)P(\omega _2)+P(\omega _1)P(\omega _3)
  +P(\omega _2)P(\omega _3)
\end{align*}
The map $P\colon\ascript\to\pscript (\ascript ^*)$ given by $A\mapsto P(A)$ is called the
\textit{master observable}. By convention $P(\emptyset )=0$ and one can verify that $P(A)P(B)=P(B)P(A)$ for all
$A,B\in\ascript$.

In general, $P(\cdot )$ is not additive or multiplicative. For example letting $A=\brac{\omega _1}$,
$B=\brac{\omega _2}$ we have that
\begin{align*}
P(AB)&=0\ne\omega _1^*\omega _2^*=P(A)P(B)\\
\intertext{and}
P(A+B)&=P\paren{\brac{\omega _1,\omega _2}}=\omega _1^*+\omega _2^*+\omega _1^*\omega _2^*\\
&\ne\omega _1^*+\omega _2^*=P(A)+P(B)
\end{align*}
As usual, we call a function $f\colon\Omega\to\real$ a \textit{random variable}. A random variable corresponds to a measurement applied to the physical system described by $(\Omega ,\ascript )$. Denoting the Borel algebra of subsets of $\real$ by $\bscript (\real )$, for a random variable $f$, we define
$P^f\colon\bscript (\real )\to\pscript (\ascript ^*)$ by $P^f (B)=P\sqbrac{f^{-1}(B)}$. Thus, $P^f=P\circ f^{-1}$ and we call $P^f$ the \textit{observable corresponding} to $f$. The next result summarizes the properties of $P(\cdot )$.

\begin{thm}       
\label{thm44}
{\rm (a)}\enspace $P(A\cup B)=P(A)\vee P(B)=P(A)+P(B)+P(A)P(B)$ for all $A,B\in\ascript$.
{\rm (b)}\enspace $P(\cdot )$ is unital, that is, $P(\Omega )=I$.
{\rm (c)}\enspace $P(\cdot )$ is strongly monotone, that is, $P(A)\le P(B)$ if and only if $A\subseteq B$.
{\rm (d)}\enspace $P(\cdot )$ is grade-2 additive, that is,
\begin{equation*}
P(A\cupdot B\cupdot C)=P(A\cupdot B)+P(A\cupdot C)+P(B\cupdot C)+P(A)+P(B)+P(C)
\end{equation*}
\end{thm}
\begin{proof}
(a)\enspace Letting $A=\brac{\alpha _1,\ldots ,\alpha _r}$, $B=\brac{\beta _1,\ldots ,\beta _s}$ we have
\begin{align*}
P(A)P(B)
  &=\sqbrac{P(\alpha _1)+\cdots +P(\alpha _r)+P(\alpha _1)P(\alpha _2)+\cdots +P(\alpha _{r-1})P(\alpha _r)}\\
  &\quad\cdot\sqbrac{P(\beta _1)+\cdots +P(\beta _s)+P(\beta _1)P(\beta _2)+\cdots +P(\beta _{s-1})P(\beta _s)}\\
  &=\sum\brac{P(\alpha _i)\colon\alpha _i\in AB}\\
  &\quad +\sum\brac{P(\alpha _i)P(\beta _j)\colon\alpha _i\in AB',\beta _j\in A'B}
\end{align*}
It follows that
\begin{align*}
P(A&\cup B)\\
  &=\sum\brac{P(\omega _i)\colon\omega _i\in A\cup B}
  +\sum\brac{P(\omega _i)P(\omega _j)\colon\omega _i,\omega _j\in A\cup B,i<j}\\
  &=P(A)+P(B)+P(A)P(B)
\end{align*}
Applying Theorem~\ref{thm41}(c) we conclude that $P(A\cup B)=P(A)\vee P(B)$.\newline
(b)\enspace Since
\begin{equation*}
P(\Omega )=\sum _{i=1}^nP(\omega _i)+\sum _{i<j=1}^nP(\omega _i)P(\omega _j)
\end{equation*}
we see that
$P(\Omega )\omega _i^*=\omega _i$ for $i=1,\ldots ,n$ and
$P(\Omega )\omega _i^*\omega _j^*=\omega _i^*\omega _j^*$ for $i,j=1,\ldots ,n$. Hence, $P(\Omega )\phi =\phi$ for all $\phi\in\ascript ^*$ so $P(\Omega )=I$.
(c)\enspace If $A\subseteq B$, then by (a) we have
\begin{equation*}
P(B)=P(A\cup B)=P(A)+P(B)+P(A)P(B)
\end{equation*}
Hence, $P(A)P(B)=P(B)P(A)=P(A)$ so $P(A)\le P(B)$. Conversely, assume that $P(A)\le P(B)$ where
$A=\brac{\alpha _1,\ldots ,\alpha _r}$, $B=\brac{\beta _1,\ldots ,\beta _s}$. Then 
\begin{align*}
P(A)&=\alpha _1^*+\cdots +\alpha _r^*+\alpha _1^*\alpha _2^*+\cdots +\alpha _{r-1}^*\alpha _r^*=P(A)P(B)\\
  &=\paren{\alpha _1^*+\cdots +\alpha _r^*+\alpha _1^*\alpha _2^*+\cdots +\alpha _{r-1}^*\alpha _r^*}\\
  &\quad\paren{\beta _1^*+\cdots +\beta _s^*+\beta _1^*\beta _2^*+\cdots +\beta _{s-1}^*\beta _s^*}
\end{align*}
If $\alpha _i\not\in B$ for some $i=1,\ldots ,r$, then $\alpha _i^*$ cannot appear in the product on the right side which is a contradiction. Hence, $\alpha _i\in B$ for $i=1,\ldots ,r$, so $A\subseteq B$.
(d)\enspace By (a) we have that
\begin{align*}
P\paren{A\cupdot B\cupdot C}&=P(A)+P\paren{B\cupdot C}+P(A)P\paren{B\cupdot C}\\
  &=P(A)+P(B)+P(C)+P(B)P(C)\\
  &\quad +P(A)\sqbrac{P(B)+P(C)+P(B)P(C)}\\
  &=P(A)+P(B)+P(C)+P(B)P(C)+P(A)P(B)\\
  &\quad +P(A)P(C)+P(A)P(B)P(C)
\end{align*}
Since $A$, $B$ and $C$ are mutually disjoint, we have that $P(A)P(B)P(C)=0$. Hence, by (a) again, we conclude that
\begin{align*}
P&\paren{A\cupdot B}+P\paren{A\cupdot C}+P\paren{B\cupdot C}+P(A)+P(B)+P(C)\\
  &=P(A)+P(B)+P(A)P(B)+P(A)+P(C)+P(A)P(C)\\
  &\quad +P(B)+P(C)+P(B)P(C)+P(A)+P(B)+P(C)\\
  &=P\paren{A\cupdot B\cupdot C}
\end{align*}
Hence, $P(\cdot )$ is grade-2 additive.
\end{proof}

\begin{exam}{1}    
Letting $\Omega =\brac{\omega _1,\omega _2}$, an ordered basis for $\ascript ^*$ is
$\omega _1^*,\omega _2^*,\omega _1^*\omega _2^*$. In terms of this basis we have
\begin{equation*}
P(\omega _1)=\left[\begin{matrix}\noalign{\smallskip}1&0&0\\
  \noalign{\smallskip}0&0&0\\\noalign{\smallskip}0&1&1\\\end{matrix}\right]\quad
  P(\omega _2)=\left[\begin{matrix}\noalign{\smallskip}0&0&0\\
  \noalign{\smallskip}0&1&0\\\noalign{\smallskip}1&0&1\\\end{matrix}\right]\quad
  P(\omega _1)P(\omega _2)=\left[\begin{matrix}\noalign{\smallskip}0&0&0\\
  \noalign{\smallskip}0&0&0\\\noalign{\smallskip}1&1&1\\\end{matrix}\right]\quad
 \end{equation*}
and $P(\omega _1)+P(\omega _2)+P(\omega _1)P(\omega _2)=I$. Projections need not commute. For instance, let $Q$ be the projection
\begin{equation*}
Q=\left[\begin{matrix}\noalign{\smallskip}1&0&0\\
  \noalign{\smallskip}0&0&0\\\noalign{\smallskip}0&0&0\\\end{matrix}\right]
 \end{equation*}
Then $QP(\omega _2)=0$ but
\begin{equation*}
P(\omega _2)Q=\left[\begin{matrix}\noalign{\smallskip}0&0&0\\
  \noalign{\smallskip}0&0&0\\\noalign{\smallskip}1&0&0\\\end{matrix}\right]
 \end{equation*}
\end{exam}

It follows from Theorem~\ref{thm44} that if $f\colon\Omega\to\real$ is a random variable, then the corresponding observable $P^f\colon\bscript (\real )\to\pscript (\ascript ^*)$ is unital, strongly monotone, grade-2 additive and satisfies $P^f(A\cup B)=P^f(A)\vee P^f(B)$ for all $A,B\in\bscript (\real )$. Our observable terminology is at odds with the usual ``intrinsic'' point of view which is observation (or measurement) independent and one can think of
$P(\cdot )$ as a mathematical construct and not refer to it as an observable.

\section{Preclusion} 
For a physical system described by $(\Omega ,\ascript )$ there are frequently theoretical or experimental reasons for excluding certain sets $A,B,\ldots\in\ascript$ from consideration. Such sets are said to be \textit{precluded}. For example, one may have an underlying quantum measure $\mu$ on $\ascript$ and sets of measure zero
($\mu (A)=0$) or sets of small measure ($\mu (A)\approx 0$) may be precluded \cite{gt09, sor94, sor07, sor09}. In a physically realistic situation, precluded events should not occur. By convention we assume that $\emptyset$ is precluded.

Let $\ascript _p\subseteq\ascript$ be the set of precluded events. We say that a coevent $\phi\in\ascript ^*$ is
\textit{preclusive} if $\phi (A)=0$ for all $A\in\ascript _p$. The set of preclusive coevents form a subspace
$\ascript _p^*$ of $\ascript ^*$ and considering $\ascript _p^*$ gives an important way of reducing the possible realities for a physical system \cite{gt09, sor07, sor09}. We now present another way of reducing the possible realities. If $\ascript _p=\brac{A_1,\ldots ,A_m}$ we say that $\phi\in\ascript ^*$ is \textit{precluding} if
$P(A_1\cup\cdots\cup A_m)\phi =0$. Thus, $\phi$ is precluding if and only if $\phi$ is in the null space of
$P(A_1\cup\cdots\cup A_m)$. This is again a subspace of $\ascript ^*$ which we will later show is contained in
$\ascript _p^*$. Applying Theorems~\ref{thm41}(c) and \ref{thm44}(a) we have that
\begin{align*}
P(A_1\cup\cdots\cup A_m)&=\vee P(A_i)=\sqbrac{\wedge P(A_i)'}'=\sqbrac{P(A_1)'\cdots P(A_m)'}'\\
  &=I+\sqbrac{I+P(A_1)}\cdots\sqbrac{I+P(A_m)}
\end{align*}
It follows that $\phi$ is precluding if and only if
\begin{equation*}
P(A_1)'\cdots P(A_m)'\phi=\phi
 \end{equation*}
Thus, the precluding coevents are precisely the coevents in the range of the projection $P(A_1)'\cdots P(A_m)'$.

\begin{thm}       
\label{thm51}
{\rm (a)}\enspace If $P(A)\phi =0$, then $\phi (A)=0$ for $A\in\ascript$, $\phi\in\ascript ^*$.\newline
{\rm (b)}\enspace If $\phi$ is precluding, the $\phi$ is preclusive.
\end{thm}
\begin{proof}
(a)\enspace Without loss of generality we can assume that $A=\brac{\omega _1,\ldots ,\omega _m}$ and that
\begin{equation*}
\phi = a_1\omega _1^*+\cdots +a_n\omega _n^*
  +b_{12}\omega _1^*\omega _2^*+\cdots +b_{n-1,n}\omega _{n-1}^*\omega _n^*
\end{equation*}
for $a_i,b_{ij}\in\integers _2$. Since
\begin{equation*}
P(A)\phi\!=\!\sqbrac{P(\omega _1)+\cdots +P(\omega _m)      
  +P(\omega _1)P(\omega _2)+\cdots +P(\omega _{m-1})P(\omega _m)}\phi=0
 \end{equation*}
we have that $a_1=\cdots =a_m=0$, $b_{ij}=0$ for $i,j\le m$ and $a_j+b_{ij}=0$ for $j>m$, $i\le m$. We conclude that $\phi$ has the form
\begin{align*}
\phi&=a_{m+1}\omega _{m+1}+\cdots +a_n\omega _n^*+a_{m+1}\omega _{m+1}^*\omega _1^*
  +\cdots +a_{m+1}\omega _{m+1}^*\omega _m^*\\
  &\quad +a_{m+2}\omega _{m+2}^*\omega _1^*+\cdots +a_{m+2}\omega _{m+2}^*\omega _m^*
  +a_n\omega _n^*\omega _1^*+\cdots +a_n\omega _n^*\omega _m^*\\
  &\quad +b_{m+1,m+2}\omega _{m+1}^*\omega _{m+2}^*+\cdots +b_{n-1,n}\omega _{n-1}^*\omega _n^*
\end{align*}
It follows that $\phi (A)=0$.
(b)\enspace Assume that $A_p=\brac{A_1,\ldots ,A_m}$. If $\phi$ is precluding, then
$P(A_1\cup\cdots\cup A_m)\phi =0$. By Theorem~\ref{thm44}(c) $P(\cdot )$ is monotone so that
\begin{equation*}
P(A_i)\phi =P(A_i)P(A_1\cup\cdots\cup A_m)\phi =0
\end{equation*}
for $i=1,\ldots ,m$. By (a) we have that $\phi (A_i)=0$, $i=1,\ldots ,m$. Hence, $\phi\in\ascript _p^*$.
\end{proof}

A \textit{precluding basis} is a set $S$ of precluding coevents such that every precluding coevent is a sum of elements of $S$. The definition of a preclusive basis is similar. Although such bases are not unique, they give an efficient way of describing all precluding (or preclusive) coevents.

\begin{exam}{2}    
Let $\Omega =\brac{\omega _1,\omega _2,\omega _3}$ and
$\ascript _p=\brac{\emptyset ,\brac{\omega _1,\omega _2}}$. It is easy to check that a preclusive basis consists of
$\omega _3^*,\omega _1^*\omega _3^*,\omega _2^*\omega _3^*,\omega _1^*+\omega _2^*,\omega _1^*
  +\omega _1^*\omega _2^*$. To find the precluding coevents we let $A=\brac{\omega _1,\omega _2}$ and solve the equation $P(A)\phi =0$. Thus,
\begin{align*}
&\sqbrac{P(\omega _1)+P(\omega _2)+P(\omega _1)P(\omega _2)}(a\omega _1^*+b\omega _2^*+c\omega _3^*
  +d\omega _1^*\omega _2^*+e\omega _1^*\omega _3^*+f\omega _2^*\omega _3^*)\\
  &\quad =a\omega _1^*+b\omega _1^*\omega _2^*+c\omega _1^*\omega _3^*+d\omega _1^*\omega _2^*
  +e\omega _1^*\omega _3^*+a\omega _1^*\omega _2^*+b\omega _2^*\\
  &\qquad +c\omega _2^*\omega _3^*+d\omega _1^*\omega _2^*+f\omega _2^*\omega _3^*
  +a\omega _1^*\omega _2^*+b\omega _1^*\omega _2^*+d\omega _1^*\omega _2^*)\\
  &\quad =0
\end{align*}
It follows that $a=b=d=0$, $c+e=c+f=0$. Hence,
\begin{equation*}
\phi =c\omega _3^*+c\omega _1^*\omega _3^*+c\omega _2^*\omega _3^*
\end{equation*}
so the only nonzero precluding coevent is
\begin{equation*}
\phi =\omega _3^*+\omega _1^*\omega _3^*+\omega _2^*\omega _3^*
\end{equation*}
Of course, $\phi$ is a precluding basis. Notice that $\phi$ is unital. This example shows that preclusive coevents need not be precluding.
\end{exam}

\begin{exam}{3}    
Let $\Omega =\brac{\omega _1,\omega _2,\omega _3}$ and
$\ascript _p=\brac{\emptyset ,\brac{\omega _1},\brac{\omega _2}}$. A preclusive basis consists of
$\omega _3^*,\omega _1^*\omega _3^*,\omega _2^*\omega _3^*,\omega _1^*\omega _2^*$. The only nonzero precluding coevent is $\phi$ obtained in Example~2. This is because
\begin{equation*}
A=\brac{\omega _1,\omega _2}=\brac{\omega _1}\cup\brac{\omega _2}
\end{equation*}
\end{exam}

\begin{exam}{4}    
Let $\Omega =\brac{\omega _!,\omega _2,\omega _3}$ and let $\ascript _p=\brac{\emptyset ,A}$ where
$A=\brac{\omega _1}$. To find the precluding coevents we solve the equation $P(A)\phi=0$. Thus,
\begin{align*}
P&(\omega _1)\sqbrac{a\omega _1^*+b\omega _2^*+c\omega _3^*+d\omega _1^*\omega _2^*
  +e\omega _1^*\omega _3^*+f\omega _2^*\omega _3^*}\\
  &\quad =a\omega _1^*+b\omega _1^*\omega _2^*+c\omega _1^*\omega _3^*+d\omega _1^*\omega _2^*
  +e\omega _1^*\omega _3^*=0
\end{align*}
Hence, $a=b+d=c+e=0$. We conclude that $\phi$ has the form
\begin{equation*}
\phi =b\omega _2^*+c\omega _3^*+b\omega _1^*\omega _2^*+c\omega _1^*\omega _3^*
  +f\omega _2^*\omega _3^*
\end{equation*}
Thus, a precluding basis consists of $\omega _2^*\omega _3^*$, $\omega _2^*+\omega _1^*\omega _2^*$ and
$\omega _3^*+\omega _1^*\omega _3^*$. The last two are not unital but sums with $\omega _2^*\omega _3^*$ are unital.
\end{exam}

We now discuss events $B$ that can actually occur. That is there exists a preclusive or precluding coevent $\phi$ such that $\phi (B)=1$. It would be nice if whenever $B$ is not precluded, then such a $\phi$ exists. But this is asking too much as simple examples show. However, we do have the following result which gives a kind of duality between preclusive and precluding coevents.

\begin{thm}       
\label{thm52}
Let $\ascript _p=\brac{A_1,\ldots ,A_m}$, $A=A_1\cup\cdots\cup A_m$ and $B\in\ascript$.\newline
{\rm (a)}\enspace If $BA'\ne\emptyset$ then there exists a preclusive coevent $\phi$ such that $\phi (B)=1$.
{\rm (b)}\enspace If there exists a precluding coevent $\phi$ such that $\phi (B)=1$, then $BA'\ne\emptyset$.
\end{thm}
\begin{proof}
(a)\enspace If $\omega\in BA'$, then $\omega\in BA'_i$, $i=1,\ldots ,m$. Hence, $\omega ^*(B)=1$ and
$\omega ^*(A_i)=0$, $i=1,\ldots ,m$. We conclude that $\omega ^*$ is preclusive.
(b)\enspace Suppose $BA'=\emptyset$. Then $B\subseteq A$ and if $\phi$ is precluding, then $P(A)\phi =0$. Hence, by Theorem~\ref{thm44}(c) we have
\begin{equation*}
P(B)\phi =P(B)P(A)\phi =0
\end{equation*}
Applying Theorem~\ref{thm51}(a), we conclude that $\phi (B)=0$. Hence, there is no precluding coevent $\phi$ such that $\phi (B)=1$. We have thus proved the contrapositive of (b) so (b) holds.
\end{proof}

\begin{cor}       
\label{cor53}
Let $\ascript _p=\brac{A_1,\ldots ,A_m}$, $A=A_1\cup\cdots A_m$ and $B\in\ascript$.\newline
{\rm (a)}\enspace If $\phi (B)=0$ for every preclusive $\phi$, then $B\subseteq A$.
{\rm (b)}\enspace If $B\subseteq A$, then $\phi (B)=0$ for every precluding $\phi$.
\end{cor}

The result in Theorem~\ref{thm52}(a) does not hold if preclusive is replaced by precluding. In Example~4,
$A_1=\brac{\omega _1}$ is the only nonempty precluded event. Letting $B=\brac{\omega _1,\omega _2}$ we have that $BA'=\brac{\omega _2}\ne\emptyset$. However, all the precluding coevents listed in Example~4 vanish on $B$. Hence, $\phi (B)=0$ for all precluding coevents. The result in Theorem~\ref{thm52}(b) does not hold if precluding is replaced by preclusive. In Example~2, letting $A=\brac{\omega _1,\omega _2}$ and
$B=\brac{\omega _1}$, $\phi =\omega _1^*+\omega _2^*$ is preclusive and $\phi (B)=1$. However,
$BA'=\emptyset$.

Examples~2 and 3 have the pleasant feature that there is a unique nonzero precluding coevent. However, the next example shows that there can be many preclusive coevents and no nonzero precluding coevent.

\begin{exam}{5}    
In the three-slit experiment $\Omega =\brac{\omega _1,\omega _2,\omega _3}$ considered previously, suppose
$\brac{\omega _1,\omega _2}$ and $\brac{\omega _2,\omega _3}$ are the only nonempty precluded events. Since
\begin{equation*}
\Omega =\brac{\omega _1,\omega _2}\cup\brac{\omega _2,\omega _3}
\end{equation*}
we have that $\phi$ is precluding if and only if $\phi =P(\Omega )\phi =0$ so the only precluding coevent is $0$. However, there are many preclusive coevents. For example,
$\omega _1^*+\omega _2^*+\omega _3^*$, $\omega _1^*+\omega _1^*\omega _2^*$,
$\omega _3^*+\omega _2^*\omega _3^*$, $\omega _1^*\omega _3^*$ form a preclusive basis.
\end{exam}

It should be mentioned that in previous works it has usually been assumed that the union of mutually disjoint precluded events is precluded. However, we did not make this assumption here.\vskip 1.5pc

\noindent\textbf{\large Acknowledgement.}\enspace
The author thanks R.~Sorkin for reading a preliminary version of this manuscript, correcting various errors and suggesting improvements.

\end{document}